\AtBeginDocument{%
  \paperwidth=\dimexpr
    1in + \oddsidemargin
    + \textwidth
    + 1in + \oddsidemargin
  \relax
  \paperheight=\dimexpr
    1in + \topmargin
    + \headheight + \headsep
    + \textheight
    + 1in + \topmargin
  \relax
  \usepackage[pass]{geometry}\relax
}

\RequirePackage{fix-cm}
\documentclass[smallextended]{svjour3}       
\smartqed  
\usepackage{graphicx}
\usepackage{subfiles,float} 
\makeatletter
\let\MYcaption\@makecaption
\makeatother

\usepackage{subcaption}
\makeatletter
\let\@makecaption\MYcaption
\makeatother
\usepackage{mathptmx}    
\usepackage{amsfonts,mathtools}
\usepackage{bm}

\begin{document}






\ \vspace{0.2cm}

{\fontsize{14}{18}\bf\noindent Eigenvalues of two-phase quantum walks with \\ one defect in one dimension}
\ \vspace{0.7cm}
\\
{\large \bf  Chusei Kiumi\ $\cdot$ \   Kei\  Saito}
\footnote{
\hspace{-6mm}
Chusei Kiumi \\
Mathematical Science Unit, Graduate School of Engineering Science, Yokohama National University, Hodogaya, Yokohama, 240-8501, Japan\\ E-mail: kiumi-chusei-bf@ynu.jp
\\ \ \\[-5pt]
Kei Saito \\
Department of Information Systems Creation, Faculty of Engineering, Kanagawa University, 
Kanagawa, Yokohama, 221-8686, Japan\\ E-mail: ksaito55.76@gmail.com
}
\vspace{1.2cm}
\begin{abstract}
 We study space-inhomogeneous quantum walks (QWs) on the integer lattice which we assign three different coin matrices to the positive part, the negative part, and the origin, respectively. We call them two-phase QWs with one defect. They cover one-defect and two-phase QWs, which have been intensively researched.  Localization is one of the most characteristic properties of QWs, and various types of two-phase QWs with one defect exhibit localization. Moreover, the existence of eigenvalues is deeply related to localization. In this paper, we obtain a necessary and sufficient condition for the existence of eigenvalues. Our analytical methods are mainly based on the transfer matrix, a useful tool to generate the generalized eigenfunctions. Furthermore, we explicitly derive eigenvalues for some classes of two-phase QWs with one defect, and illustrate the range of eigenvalues on unit circles with figures. Our results include some results in previous studies, e.g. Endo et al. (2020).
\end{abstract}

\section{Introduction}
\label{intro}
Discrete-time quantum walks (QWs) are the quantum counterparts of classical random walks, and their applications have attracted various research fields.
Two-state QWs on $\mathbb{Z}$ are well known as the most typical models \cite{AAKVW,CGMV,K}. In particular, QWs, the dynamics of the walker depend on each position, are called space-inhomogeneous QWs, and various applications have been studied. For instance, one-defect QWs where the walker at the origin behaves differently are used in quantum searching algorithms \cite{AKR,CG,KSW}, and two-phase QWs where the walker behaves differently in each of negative and non-negative part are related to the research of topological insulators \cite{KRBD}. In this paper, we deal with two-phase QWs with one defect, which are combined models of two-phase and one-defect QWs. In particular, we concentrate on the analysis of the eigenvalues of the time evolution via the method of transfer matrices.
It is known that the existence of eigenvalues is a necessary and sufficient condition for localization, which is a characteristic property of QWs \cite{S,RST1,RST2}.
The method of transfer matrices is used by \cite{KKK,KKK2,MSSSS} to construct generalized eigenfunctions and stationary measures. Endo \cite{EKST2} treated the specific case of two-phase QWs with one defect and succeeded in revealing the eigenvalues and time averaged limit measures. However, the results for general cases have not been clarified.

The rest of this paper is organized as follows.
In Section \ref{sec:2}, we introduce some notations and definitions of our models.
Moreover, we present Theorem \ref{THEO_1}, which gives a necessary and sufficient condition for the existence of eigenvalues. This theorem is a key result of this paper. In Section \ref{sec:3}, we use Theorem \ref{THEO_1} to clarify the eigenvalues for the five classes of two-phase QWs with one defect, where four of them include previous studies \cite{EEKK,EKO,EKST,EKST2,WLKGGB}. 
The remaining section is assigned to summarize our results.

\section{Definitions and techniques  }
\label{sec:2}
\subsection{Quantum walks on $\mathbb{Z}$}
Firstly, we introduce two-state QWs on the integer lattice $\mathbb{Z}$.
Let $\mathcal{H}=\ell^2(\mathbb{Z} ; \mathbb{C}^2)$ be a Hilbert space of states. The time evolution $U=SC$ is defined by the product of two unitary operators $S$ and $C$, on $\mathcal{H}$. Here, $S$ denotes the shift operator, and $C$ denotes the coin operator, respectively.
For $\Psi\in\mathcal{H}$, we write $\Psi(x)={}^T\begin{bmatrix}\Psi_L(x) & \Psi_R(x)\end{bmatrix}$, then $S$ and $C$ are defined as follows:
 \begin{align*}
	(S\Psi)(x)=\begin{bmatrix}
	\Psi_L(x+1)	\\ \Psi_R(x-1)
	\end{bmatrix},
	\qquad
	(C\Psi)(x)
	=
	\tilde C_x\Psi(x),
	\end{align*}where $\{\tilde C_x\}_{x\in\mathbb{Z}}$ is a sequence of $2\times2$ unitary matrices called coin matrices. We write $\tilde C_x$ as below:
    \begin{align*}
	\tilde C_x = e^{i\Delta_x}
	\begin{bmatrix}
	\alpha_x & \beta_x
	\\
	-\overline{\beta_x} & \overline{\alpha_x}
	\end{bmatrix},
	\end{align*}where $\alpha_x, \beta_x\in\mathbb{C},\ \alpha_x\neq 0,\ \Delta_x\in [0,2\pi)$ and $|\alpha_x|^2+|\beta_x|^2=1$. We define two matrices $P_x$ and $Q_x$ by
 \begin{align*}
	P_{x}=\begin{bmatrix}
	1 & 0 \\ 0 & 0
	\end{bmatrix}\tilde C_x,
	\qquad
	Q_{x}=\begin{bmatrix}
	0 & 0 \\ 0 & 1
	\end{bmatrix}\tilde C_x.
	\end{align*}Then, the time evolution can be written as 
    \begin{align}
	\label{A01_a}
	(U\Psi)(x)=P_{x+1}\Psi(x+1)+Q_{x-1}\Psi(x-1), \quad x\in\mathbb{Z}.
	\end{align}

\begin{remark}
If there exists $x_0\in\mathbb{Z}$ such that $\alpha_{x_0}=0$, then
\begin{align*}
	P_{x_0}P_{x_0+1}=Q_{x_0}Q_{x_0-1}=O
	\end{align*}holds, where $O$ denotes the zero matrix. Thus, position $x_0$ becomes reflecting boundary, and $\mathbb{Z}$ is divided into disconnected half-lines. In this paper, we do not treat this case. Hence, we assume $\alpha_x\neq0$ for all $x\in\mathbb{Z} $.
    
\end{remark}

Let $\Psi_0\in\mathcal{H}\ (\|\Psi_0\|^2=1)$ be an initial state of the QW, then the probability distribution at time $t\in\mathbb{Z}_{\geq 0}$ is defined by  $\mu_t^{(\Psi_0)}(x)=\|(U^t\Psi_0)(x)\|^2$. Here, we say that the QW exhibits localization if there exists an initial state $\Psi_0\in\mathcal{H}$ and position $x_0\in\mathbb{Z}$ which satisfy $\limsup_{t\to\infty}\mu^{(\Psi_0)}_t(x)>0$. In particular, it is known that the QW exhibits localization if and only if there exists an eigenvalue of $U$, namely, there exists $\lambda\in[0, 2\pi),\ \Psi\in\mathcal{H}$  such that
\begin{align*}
	U\Psi=e^{i\lambda}\Psi.
	\end{align*}
	
	\begin{theorem}
	\label{THEO_rangeEv}
	If there exists a constant $N$ such that the set of parameters $(\alpha_x,\ \beta_x, \Delta_x)$ becomes a set of constants $(\alpha,\beta,\Delta)$ for $|x|>N$, then $|\cos(\lambda - \Delta)|>|\alpha|$ holds for any $e^{i\lambda}\in\sigma_{\rm p}(U)$.
	Here, $\sigma_{\rm p}(U)$ denotes the set of eigenvalues of $U$.
	\end{theorem}
	\begin{proof}
	See \cite{MSSSS}.
    \qed
	\end{proof}
    Let $J$ be a unitary operator on $\mathcal{H}$ defined as	\begin{align*}
	    (J\Psi)(x)=
	    \begin{bmatrix}
	    \Psi_L(x-1) \\ \Psi_R(x)
	    \end{bmatrix},
	    \quad \Psi\in\mathcal{H},\ x\in\mathbb{Z}.
	\end{align*}
	The inverse of $J$ is given as
\begin{align*}
	   	(J^{-1}\Psi)(x)=
	    \begin{bmatrix}
	    \Psi_L(x+1) \\ \Psi_R(x)
	    \end{bmatrix},
	    \quad \Psi\in\mathcal{H},\ x\in\mathbb{Z}.
	\end{align*}
	Moreover, we introduce a transfer matrix $T_x(\lambda)$ for $\lambda\in[0,2\pi),\ x\in\mathbb{Z}$ as below:
	\begin{align*}
	T_x(\lambda)=
	\frac{1}{\alpha_x}
	\begin{bmatrix}
	e^{i(\lambda-\Delta_x)} & -\beta_x
	\\
	-\overline{\beta_x} & e^{-i(\lambda-\Delta_x)}
	\end{bmatrix}.
	\end{align*}The transfer matrix is a normal matrix and its inverse matrix is
	\begin{align*}
	T_x^{-1}(\lambda)=
	\frac{\alpha_x}{|\alpha_x|^2}
	\begin{bmatrix}
	e^{-i(\lambda-\Delta_x)} & \beta_x
	\\
	\overline{\beta_x} & e^{i(\lambda-\Delta_x)}
	\end{bmatrix}.
	\end{align*}These definitions lead us to the key proposition for this paper.

	\begin{proposition}
	\label{PROP_Transfer}
	    For $\lambda\in[0,2\pi)$ and $\Psi\in\mathcal{H}$, following (i) and (ii) are equivalent.
	    \begin{align*}
	       &\text{(i)}\quad \Psi\in\ker(U-e^{i\lambda}).
	       \\
	       &\text{(ii)}\quad (J\Psi)(x+1)=T_x(\lambda)(J\Psi)(x),\quad x\in\mathbb{Z}.
	    \end{align*}

	\end{proposition}
	\begin{proof}
    $\Psi\in\ker(U-e^{i\lambda})$ if and only if the following conditions hold for  $x\in \mathbb{Z}$:
    	\begin{align*}
	e^{i\lambda}\Psi_L(x-1) &= e^{i\Delta_{x}}(\alpha_{x}\Psi_L(x)+\beta_{x}\Psi_R(x)),
	\\
		e^{i\lambda}\Psi_R(x+1) &= e^{i\Delta_{x}}(-\overline{\beta_{x}}\Psi_L(x)+\overline{\alpha_{x}}\Psi_R(x)).
	\end{align*}
    By a direct calculation, we get
    \begin{align*}
	\begin{bmatrix}
	e^{i\Delta_x}\alpha_x & 0
	\\
	-e^{i\Delta_x}\overline{\beta_x} & -e^{i\lambda}
	\end{bmatrix}
	(J\Psi)(x+1)
	=
	\begin{bmatrix}
	e^{i\lambda} & -e^{i\Delta_x}\beta_x
	\\
	0 & -e^{i\Delta_x}\overline{\alpha_x}
	\end{bmatrix}
	(J\Psi)(x)
	\end{align*}which is equivalent to 
    \begin{align*}
    (J\Psi)(x+1)=
    \frac{1}{\alpha_x}
	\begin{bmatrix}
	e^{i(\lambda-\Delta_x)} & -\beta_x
	\\
	-\overline{\beta_x} & e^{-i(\lambda-\Delta_x)}
	\end{bmatrix}(J\Psi)(x).
    \end{align*}\qed
	\end{proof}
	\begin{corollary}
	\label{CORO_A}
	For $\lambda\in[0,2\pi)$ and $\varphi\in\mathbb{C}^2\setminus\{\bm{0}\}$, we define $\tilde\Psi : \mathbb{Z}\to \mathbb{C}^2$ as follows:
	\begin{align}
	\label{A01_b}
	\tilde\Psi(x)=
	\begin{cases}
	T_{x-1}(\lambda)T_{x-2}(\lambda)\cdots T_1(\lambda) T_0(\lambda)\varphi,\quad &x>0,
	\\[+5pt]
	\varphi, \quad & x=0,
	\\[+5pt]
	T_{x}^{-1}(\lambda)T_{x+1}^{-1}(\lambda)\cdots T_{-2}^{-1}(\lambda)T_{-1}(\lambda)\varphi,\quad &x<0.
	\end{cases}
	\end{align}
    Then, the existence of $\varphi$ satisfying $\tilde\Psi\in\mathcal{H}$ provides a necessary and sufficient condition for $e^{i\lambda}\in\sigma_{\rm p}(U)$.
    \end{corollary}
	\begin{proof}
	Firstly, for $\lambda \in [0,2\pi),$ if there exists $\varphi$ satisfying $\tilde\Psi\in\mathcal{H}$, then $\Psi = J^{-1}\tilde \Psi$ satisfies (ii) in Proposition \ref{PROP_Transfer}, therefore $e^{i\lambda}\in\sigma_{\rm p}(U)$.
	Secondly, if there exists $\Psi\in\ker(U-e^{i\lambda})$, then $\varphi =\Psi(0)$ provides $\tilde \Psi\in\mathcal{H}$.
    \qed
	\end{proof}
    
    Henceforward, we discuss about $\varphi$ satisfying $\tilde\Psi\in\mathcal{H}$. Let $\zeta_{x,\pm}$ be eigenvalues of $T_x(\lambda)$. Then,
    \begin{align}
	\label{A02_a}
	\zeta_{x,\pm}=\frac{\cos(\lambda-\Delta_x)\pm\sqrt{\cos^2(\lambda -\Delta_x)-|\alpha_x|^2}}{\alpha_x}.
	\end{align}Additionally, their associated eigenvectors $|v_{x,\pm}\rangle$ are given as follows:

\begin{enumerate}
 \item[]{$\bullet\ \beta_x\neq 0$ case :} 
\begin{align}
\label{A03_a}
	|v_{x,\pm}\rangle
	=
	\begin{bmatrix}
	\beta_x \\ i\sin(\lambda -\Delta_x)\mp\sqrt{\cos^2(\lambda -\Delta_x)-|\alpha_x|^2}
	\end{bmatrix},
\end{align}
\item[]{$\bullet\ \beta_x = 0$ case : }
\begin{align}
\label{A04_a}
	|v_{x,+}\rangle
	=
	\begin{bmatrix}
	1 \\ 0
	\end{bmatrix}
	,
	\quad
	|v_{x,-}\rangle
	=
	\begin{bmatrix}
	0 \\ 1
	\end{bmatrix}
	.
\end{align}
\end{enumerate}	
Note that when $\beta_x=0$, $\zeta_{x,\pm}$ is described as $e^{\pm i(\lambda -\Delta_x)}/\alpha_x$.
\begin{lemma}
\label{LEM_A}
For constants $N\in\mathbb{N},\ \alpha\in\mathbb{C}\setminus\{0\}$ and $\Delta\in[0,2\pi)$, we assume $(\alpha_x, \Delta_x)=(\alpha, \Delta)$ for $x>N$ or $x<-N$.
If $|\cos(\lambda -\Delta)|\leq|\alpha|$ holds, then $e^{i\lambda}\not\in\sigma_{\rm p}(U)$.
\end{lemma}
\begin{proof}
Under the assumption, from (\ref{A02_a}) and $|\cos(\lambda -\Delta)|\leq|\alpha|$, we can see that $|\zeta_{x,\pm}|=1$ holds for $x>N$ or $x<-N$.
Then, $\tilde\Psi(x)$ defined by (\ref{A01_b}) does not converge to $\bm{0}$ since $\ker T_{x}(\lambda)=\{\bm{0}\}$ and $|\zeta_{x,+}|$ or $|\zeta_{x,-}|$ equals $1$.
Thus, $\tilde\Psi$ is not included in $\ell^2(\mathbb{Z};\mathbb{C}^2)$, and Corollary \ref{CORO_A} leads to the statement.
\qed
\end{proof}
\subsection{Two-phase quantum walks with one defect}
In this paper, we consider two-phase QWs with one defect defined as  
\begin{align}
	\label{B01_a}
	(\alpha_x, \beta_x, \Delta_x)=
	\begin{cases}
	(\alpha_m, \beta_m, \Delta_m),\quad &x<0,
	\\
	(\alpha_o, \beta_o, \Delta_o),\quad &x=0,
	\\
	(\alpha_p, \beta_p, \Delta_p),\quad &x>0,
	\end{cases}
	\end{align}
    where $\alpha_j, \beta_j\in\mathbb{C},\ \Delta_j\in[0,2\pi),\ |\alpha_j|^2+|\beta_j|^2=1$ and $\alpha_j\neq 0$ for $j\in \{p, o, m\}.$
    Similarly, we write $T_x(\lambda)=T_j(\lambda),\ \zeta_{x,\pm}=\zeta_{j,\pm},\ |v_{x,\pm}\rangle = |v_{j,\pm}\rangle$, where $j=p\ (x>0)$, $=o\ (x=0)$, $=m\ (x<0)$.
\begin{theorem}
	\label{THEO_1}
	For $\lambda\in[0,2\pi )$, $e^{i\lambda}\in\sigma_p(U)$ if and only if the condition $|\cos(\lambda-\Delta_j)|>|\alpha_j|\ (\ j\in\{p,m\})$ and $\det D(\lambda)=0$ hold. Here $D(\lambda)$ is defined by	\begin{align*}
	    D(\lambda)
	    =
	    \begin{bmatrix}
	    \langle v_{p,s_p}^\perp, v_{o,+}\rangle\zeta_{o,+} & 
	    \langle v_{p,s_p}^\perp, v_{o,-}\rangle\zeta_{o,-}
	    \\
    	\langle v_{m,s_m}^\perp, v_{o,+}\rangle & 
	    \langle v_{m,s_m}^\perp, v_{o,-}\rangle
	    \end{bmatrix},
	\end{align*}where $|v_{j,\pm}^\perp\rangle$ is a non-zero vector satisfying $\langle v_{j,\pm}^\perp ,v_{j,\pm}\rangle =0\ (j\in\{p,m\})$, and  $s_p,s_m$ are plus or minus sign determined by 
	\begin{align*}
	s_p=
	\begin{cases}
	+,\quad & -\cos(\lambda -\Delta_p)>0,
	\\
	-,\quad & -\cos(\lambda -\Delta_p)<0,
	\end{cases}
	\qquad
	s_m=
	\begin{cases}
	+,\quad & \cos(\lambda -\Delta_m)>0,
	\\
	-,\quad & \cos(\lambda -\Delta_m)<0.
	\end{cases}
	\end{align*}	In particular, if the condition holds, the multiplicity of the eigenvalue equals $1$.
	\end{theorem}
	\begin{proof}
    We show that $\det D(\lambda)=0$ is equivalent to the existence of $\varphi\in\mathbb{C}^2$ satisfying $\tilde \Psi\in\mathcal{H}$ in Corollary \ref{CORO_A}. Since we consider the two-phase QWs with one defect defined by (\ref{B01_a}), $\tilde\Psi$ is given as \begin{align}
    \label{B03_a}
    \tilde{\Psi}(x)=\left\{\begin{array}{ll}
T_{p}^{x-1}(\lambda) T_{0}(\lambda) \varphi, & x>0, \\
\varphi, & x=0, \\
T_{m}^{x}(\lambda) \varphi, & x<0.
\end{array}\right.
    \end{align}
Let $\varphi = a |v_{o,+}\rangle + b |v_{o,-}\rangle,\ a,b\in\mathbb{C}$, by the change of basis, we get
	\begin{align*}
	|v_{o,\pm}\rangle = 
	\frac{\langle v_{j,-}^\perp, v_{o,\pm}\rangle}{\langle v_{j,-}^\perp, v_{j,+}\rangle}
	|v_{j,+}\rangle
	+
	\frac{\langle v_{j,+}^\perp, v_{o,\pm}\rangle}{\langle v_{j,+}^\perp, v_{j,-}\rangle}
	|v_{j,-}\rangle
	,
	\quad
	j\in\{p,m\}.
	\end{align*}Therefore, from (\ref{B03_a}), we obtain

	\begin{align}
	\label{B04_a}
	\tilde\Psi(x)
	=
	\begin{cases}
	&\dfrac{\zeta_{o,+}\langle v_{p,-}^\perp, v_{o,+}\rangle
	\,a
	+
	\zeta_{o,-}\langle v_{p,-}^\perp, v_{o,-}\rangle
	\,b
	}
	{\langle v_{p,-}^\perp, v_{p,+}\rangle}
	\, \zeta_{p,+}^{x-1}
	|v_{p,+}\rangle
	\\[+15pt]
	&
	\hspace{1.1cm}
	+
	\dfrac{\zeta_{o,+}\langle v_{p,+}^\perp, v_{o,+}\rangle
	\,a
	+
	\zeta_{o,-}\langle v_{p,+}^\perp, v_{o,-}\rangle
	\,b
	}
	{\langle v_{p,+}^\perp, v_{p,-}\rangle}
	\, \zeta_{p,-}^{x-1}
	|v_{p,-}\rangle
	,\ x>0,
	\\[+25pt]
	&\dfrac{\langle v_{m,-}^\perp, v_{o,+}\rangle
	\, a
	+
	\langle v_{m,-}^\perp, v_{o,-}\rangle
	\, b
	}
	{\langle v_{m,-}^\perp, v_{m,+}\rangle}
	\, \zeta_{m,+}^{x}
	|v_{m,+}\rangle
	\\[+15pt]
	&
	\hspace{2cm}
	+
	\dfrac{\langle v_{m,+}^\perp, v_{o,+}\rangle
	\, a
	+
	\langle v_{m,+}^\perp, v_{o,-}\rangle
	\, b
	}
	{\langle v_{m,+}^\perp, v_{m,-}\rangle}
	\, \zeta_{m,-}^{x}
	|v_{m,-}\rangle
	,\hspace{0.3cm}x<0.
	\end{cases}
	\end{align}From Lemma \ref{LEM_A}, it is sufficient to consider the case of $|\cos(\lambda -\Delta_j)|>|\alpha_j|,\ j\in\{p,m\}.$
In this case, from (\ref{A02_a}), we can say that either $|\zeta_{j,+}|$ or $|\zeta_{j,-}|$ is greater than 1 and the other is less than 1, namely, if $\cos(\lambda -\Delta_j)>|\alpha_j|$,
then $|\zeta_{j,+}|>1$ and $|\zeta_{j,-}|<1$ hold, and if $\cos(\lambda -\Delta_j)<-|\alpha_j|$, $|\zeta_{j,+}|<1$ and $|\zeta_{j,-}|>1$ hold.
Therefore,
\begin{align*}
	    \zeta_{o,+}\langle v_{p,s_p}^\perp, v_{o,+}\rangle\, a+
	    \zeta_{o,-}\langle v_{p,s_p}^\perp, v_{o,-}\rangle\, b
	    =
    	\langle v_{m,s_m}^\perp, v_{o,+}\rangle \,a+
	    \langle v_{m,s_m}^\perp, v_{o,-}\rangle \,b =0,
	\end{align*} which is equivalent to ${}^T[a \ \    b]\in\ker D(\lambda)$ is necessary for $\tilde\Psi\in\mathcal{H}$. It is also sufficient for $\tilde\Psi\in\mathcal{H}$ since $\tilde{\Psi}$ with ${}^T[a \ \    b]\in\ker D(\lambda)$ satisfies $\|\tilde\Psi(x)\|^2=O(c^{|2x|})$ for a constant $|c|<1$. Thus, $\det D(\lambda)=0$ is a necessary and sufficient condition for $e^{i\lambda}\in\sigma_{\rm p}(U)$ from Corollary \ref{CORO_A}. Additionally, from Proposition \ref{PROP_Transfer}, eigenfunctions of $U$ are given by $J^{-1}\tilde\Psi$. Therefore, the dimension of $\ker(U-e^{i\lambda})$ is determined by the freedom of choice of $\varphi$. Since $\varphi$ is determined by $a,b$, and $D(\lambda)$ is not a zero matrix,
	 $\dim\ker(U-e^{i\lambda})=\dim\ker D(\lambda)=1$ holds if ${}^T[a \ \    b]\in\ker D(\lambda)$.
\qed
	\end{proof}

\section{Main Theorems}
\label{sec:3}
	In this section, we use Theorem \ref{THEO_1} to analyse eigenvalues for 5 models, and 4 of them include the models in the previous studies.
	\begin{theorem}
	\label{THEO_QW1}
	We assume $(\alpha_p, \beta_p, \Delta_p) = (\alpha_m, \beta_m, \Delta_m) = (\alpha, \beta, \Delta)$ and $\Delta_o = \Delta$. 
	$\sigma_p(U)\neq \emptyset$ if and only if $|\beta|^2 >\Re(\beta\overline{\beta_0})$ holds, and all eigenvalues of $U$ are given by:
	\begin{align*}
e^{i\lambda _{1}}  =\frac{\left( \Re \left(\overline{\beta _{0}} \beta \right) -1\right) +i\sqrt{|\beta |^{2} -\Re ^{2}\left(\overline{\beta _{0}} \beta \right)}}{\sqrt{1+|\beta |^{2} -2\Re \left(\overline{\beta _{0}} \beta \right)}} e^{i\Delta }
,\qquad 
e^{i\lambda _{2}}  =-e^{i\lambda _{1}},
\\
e^{i\lambda _{3}}  =\frac{\left( \Re \left(\overline{\beta _{0}} \beta \right) -1\right) -i\sqrt{|\beta |^{2} -\Re ^{2}\left(\overline{\beta _{0}} \beta \right)}}{\sqrt{1+|\beta |^{2} -2\Re \left(\overline{\beta _{0}} \beta \right)}} e^{i\Delta },\qquad 
e^{i\lambda _{4}}  =-e^{i\lambda _{3}}.
	\end{align*}
	\end{theorem}
   
	\begin{theorem}
	\label{THEO_QW2}
	We assume $(\alpha_p, \beta_p, \Delta_p) = (\alpha_m, \beta_m, \Delta_m) = (\alpha, \beta, \Delta)$ and $(\alpha_o, \beta_o) = (\alpha, \beta)$.
	$\sigma_p(U)\neq \emptyset$ if and only if Condition 1 or Condition 2 holds.
	\begin{enumerate}
	\item[]{\ $\bullet$\ \rm Condition 1 :}\quad $|\beta|\cos(\Delta_o-\Delta)-|\alpha|\sin(\Delta_o-\Delta)<|\beta|$.
	\item[]{\ $\bullet$\ \rm Condition 2 :}\quad $|\beta|\cos(\Delta_o-\Delta)+|\alpha|\sin(\Delta_o-\Delta)<|\beta|$.
	\end{enumerate}
	Then all eigenvalues of $U$ are given by the followings:
	\\[+5pt]
	\ $\bullet$\ If Condition 1 holds, $e^{i\lambda_1}, e^{i\lambda_2}\in\sigma_{\rm p}(U)$, where
	\begin{align*}
	e^{i\lambda_1}=
	\dfrac
	{|\beta|(|\beta|+i|\alpha|)e^{i\Delta_o}-e^{i\Delta}}
	{\left|
	|\beta|(|\beta|+i|\alpha|)e^{i\Delta_o}-e^{i\Delta}
	\right|},\quad e^{i\lambda_2}=-e^{i\lambda_1}.
	\end{align*}
	$\bullet$\ If Condition 2 holds, $e^{i\lambda_3}, e^{i\lambda_4}\in\sigma_{\rm p}(U)$, where
	\begin{align*}
	e^{i\lambda_3}=
	\dfrac
	{|\beta|(|\beta|-i|\alpha|)e^{i\Delta_o}-e^{i\Delta}}
	{\left|
	|\beta|(|\beta|-i|\alpha|)e^{i\Delta_o}-e^{i\Delta}
	\right|}, \quad e^{i\lambda_4}=-e^{i\lambda_3}.
	\end{align*}
	In particular, if both of Condition 1 and Condition 2 hold, $U$ has 4 eigenvalues written above.
	\end{theorem}
\begin{figure}[H]
\begin{subfigure}{0.5\textwidth}
\vspace{0.02cm}
\centering
\includegraphics[width=0.9\linewidth, height=4.5cm]{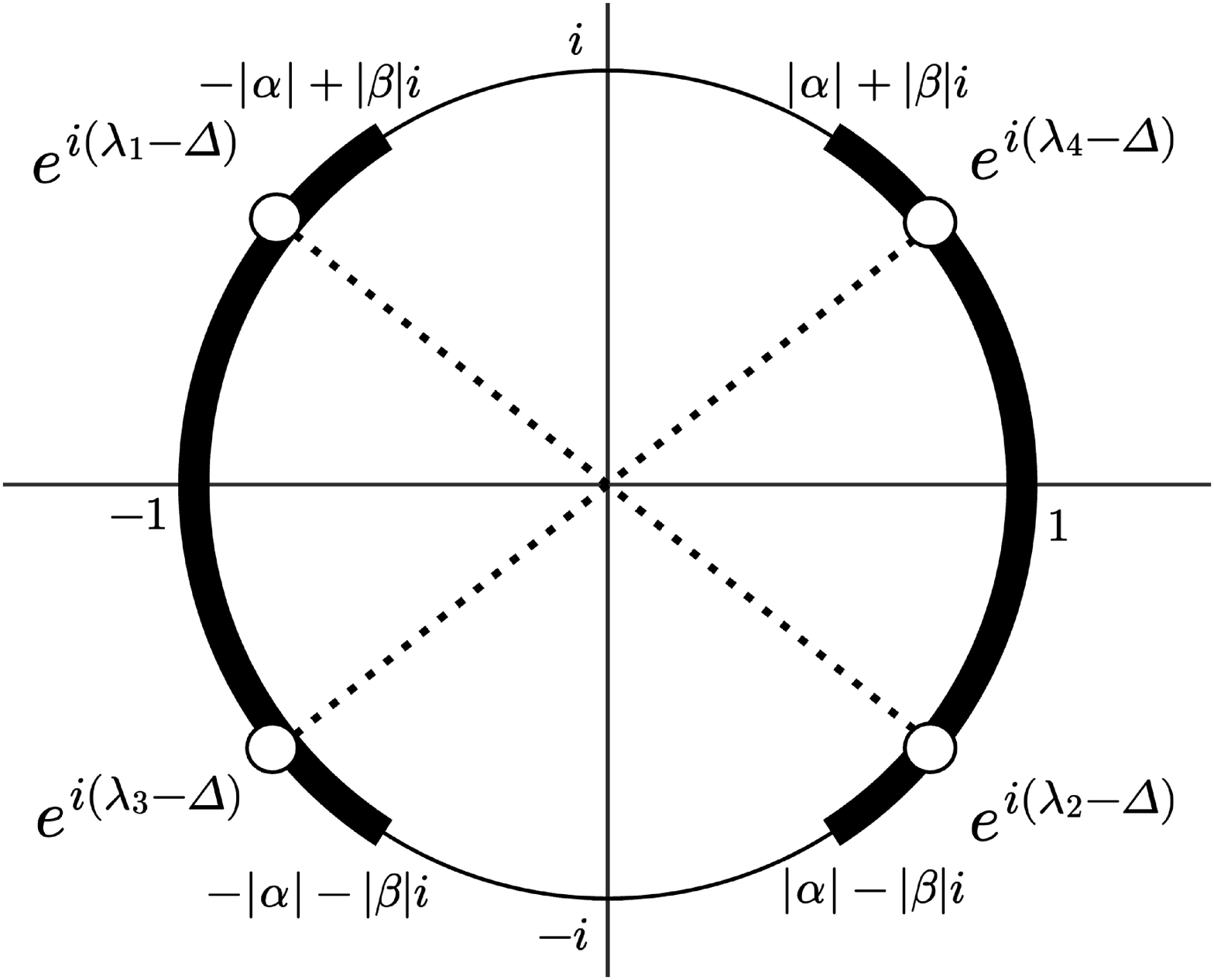} 
\caption{}
\end{subfigure}
\begin{subfigure}{0.5\textwidth}
\centering
\includegraphics[width=0.9\linewidth, height=4.5cm]{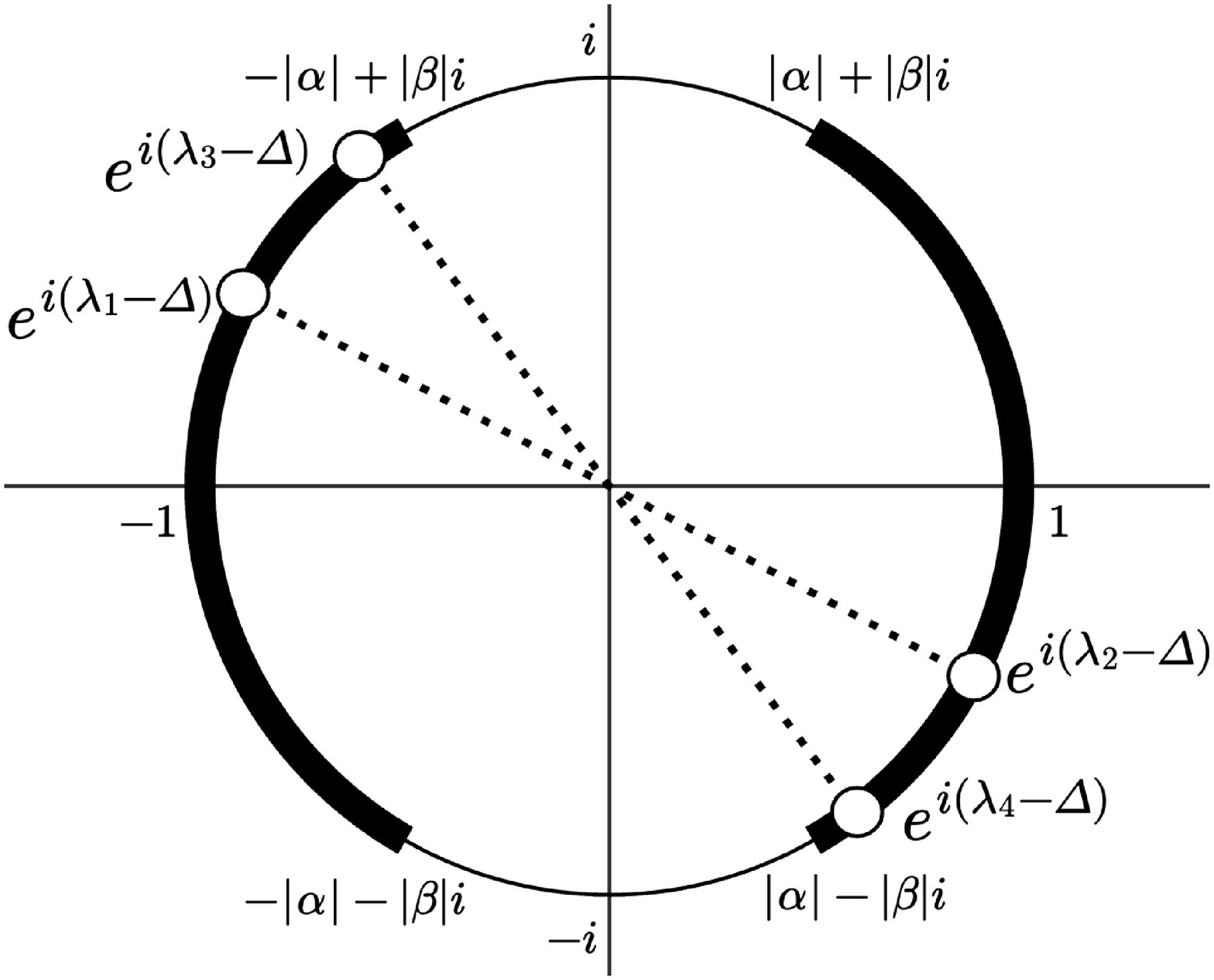}
\caption{}
\end{subfigure}
\caption{(a) and (b) illustrate eigenvalues rotated through $-\Delta$ for the one-defect QWs in Theorem \ref{THEO_QW1} and Theorem \ref{THEO_QW2}, respectively. The bold lines indicate the range of the existence of eigenvalues rotated through $-\Delta$.  }
\label{fig:1}  
\end{figure}
\begin{remark}
The models of Theorem \ref{THEO_QW1} and Theorem  \ref{THEO_QW2}  are extensions of the models in the following previous studies, respectively:
	\begin{enumerate}
	\item[]$\bullet$ Endo, Konno, Segawa, Takei (2014) \cite{EKST} : \[(\alpha,\beta,\Delta)=(\frac{i}{\sqrt{2}}, \frac{i}{\sqrt{2}}, \frac{3\pi}{2}),\qquad (\alpha_o,\beta_o,\Delta_o)=(i\cos\xi, i\sin\xi, \frac{3\pi}{2}),\quad \xi\in(0, \frac{\pi}{2}).\]
	\item[]$\bullet$
	Wojcik et al. (2012) \cite{WLKGGB},\quad 
	Endo, Konno (2014) \cite{EK} : 
	\[(\alpha,\beta,\Delta)=(\frac{i}{\sqrt{2}}, \frac{i}{\sqrt{2}}, \frac{3\pi}{2}),\quad (\alpha_o,\beta_o,\Delta_o)=(\frac{i}{\sqrt{2}}, \frac{i}{\sqrt{2}}, \frac{3\pi}{2}+2\pi\phi),\quad \phi\in(0,1).\]
	\end{enumerate}
	
	\end{remark}

\begin{theorem}
	\label{TWO_PHASE_QW1}
	We assume $(\alpha_o, \beta_o, \Delta_o)=(\alpha_p, \beta_p, \Delta_p)$ and $\arg \beta_p = \arg \beta_m$, where $\arg z$ denotes the argument of a complex number $z$.
	$\sigma_p(U)\neq \emptyset$ if and only if $\cos \left(\Delta_{m}-\Delta_{p}\right)<\left|\beta_{m}\right|\left|\beta_{p}\right|-\left|\alpha_{m}\right|\left|\alpha_{p}\right|$ holds, and all eigenvalues of $U$ are given by
	\begin{align*}
	e^{i \lambda_1}= \frac{\left|\beta_{p}\right| e^{i \Delta_{m}}-\left|\beta_{m}\right| e^{i \Delta_{p}}}{|| \beta_{p}\left|e^{i \Delta_{m}}-\right| \beta_{m}\left|e^{i \Delta_{p}}\right|},\ 
    e^{i \lambda_2}=-e^{i \lambda_1}.
	\end{align*}
	\end{theorem}

    \begin{theorem}
	\label{TWO_PHASE_QW2}
	We assume $(\alpha_o, \beta_o, \Delta_o)=(\alpha_p, \beta_p, \Delta_p)$ and $\Delta_p = \Delta_m =\Delta$. 
	$\sigma_p(U)\neq \emptyset$ if and only if $\left(\mathfrak{R}\left(\beta_{m} \overline{\beta_{p}}\right)-\left|\beta_{p}\right|^{2}\right)\left(\mathfrak{R}\left(\beta_{m} \overline{\beta_{p}}\right)-\left|\beta_{m}\right|^{2}\right)>0$ holds, and all eigenvalues of $U$ are given by
\small	\begin{align*}
 e^{i \lambda_1}=  \frac{e^{i \Delta}\left(\sqrt{\left(\Re\left(\beta_{m} \overline{\beta_{p}}\right)+\left|\alpha_{p}\right|\left|\alpha_{m}\right|-1\right)\left(\Re\left(\beta_{m} \overline{\beta_{p}}\right)-\left|\alpha_{p}\right|\left|\alpha_{m}\right|-1\right)}+i \mathfrak{I}\left(\beta_{m} \overline{\beta_{p}}\right)\right)}{\left|\beta_{p}-\beta_{m}\right|},\ e^{i \lambda_2}=-e^{i \lambda_1}
	\end{align*}
	\end{theorem}
    
      \begin{figure}[H]
\begin{subfigure}{0.5\textwidth}
\centering
\includegraphics[width=0.9\linewidth, height=4.5cm]{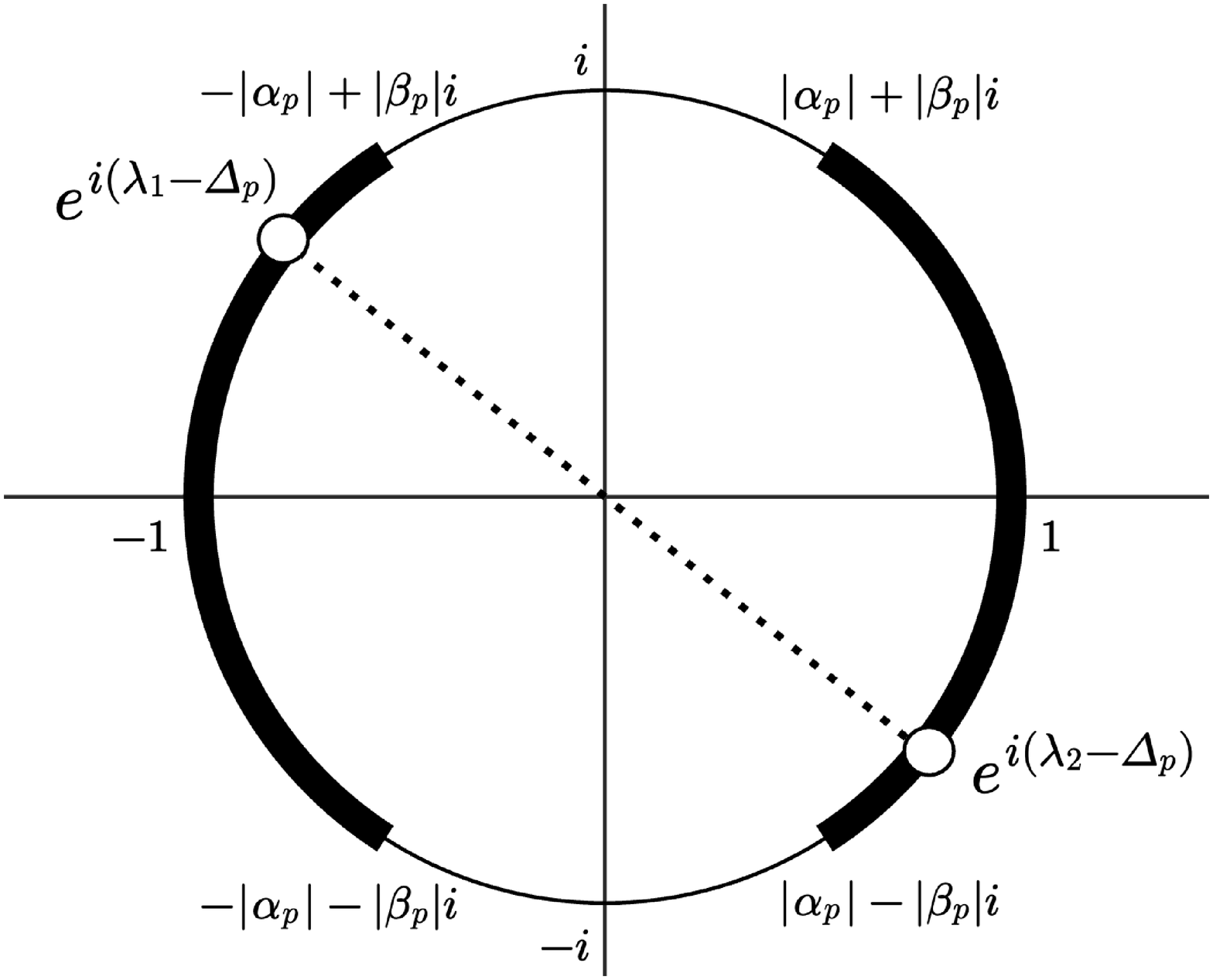} 
\caption{}
\end{subfigure}
\begin{subfigure}{0.5\textwidth}
\centering
\includegraphics[width=0.9\linewidth, height=4.5cm]{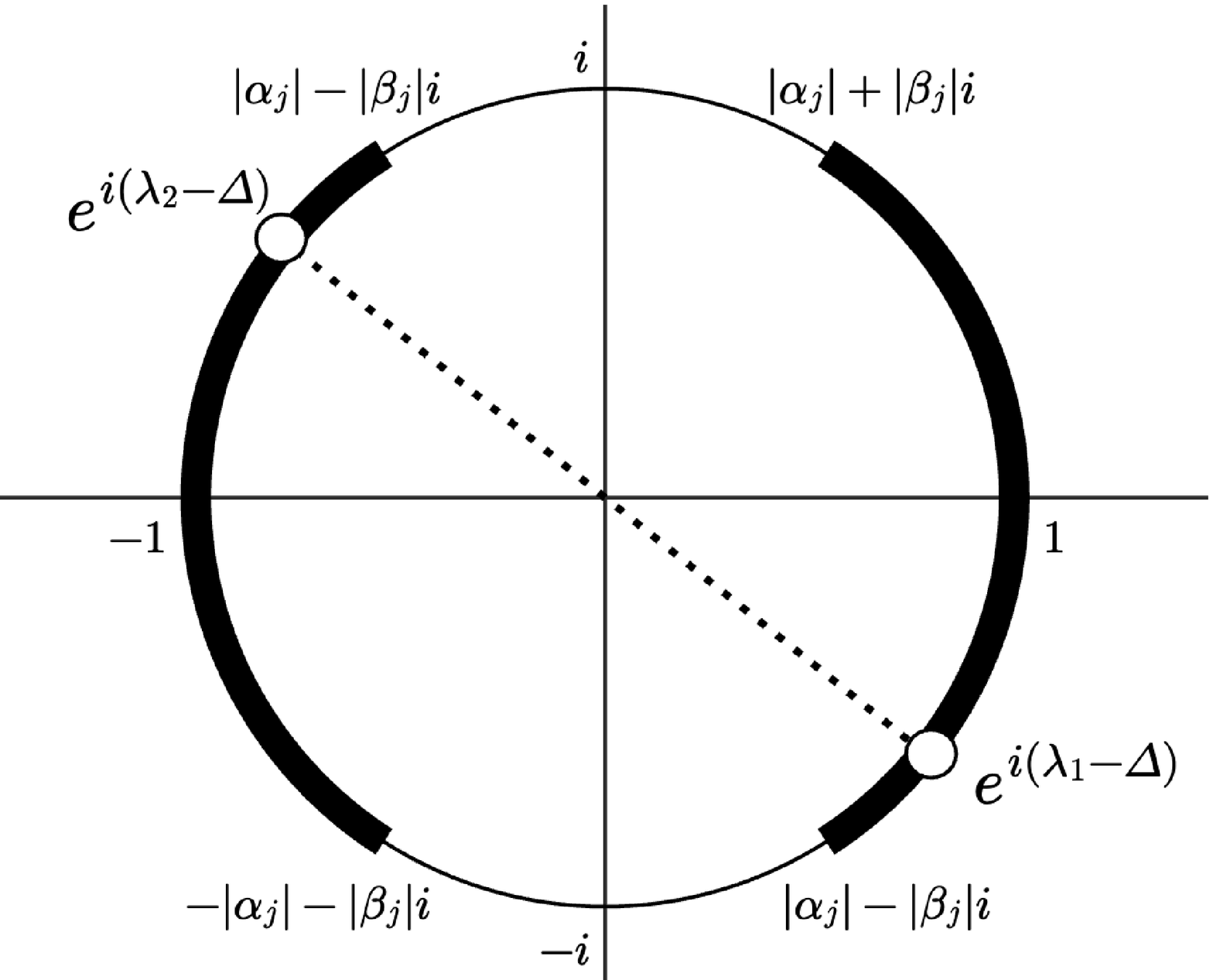}
\caption{}
\end{subfigure}
\caption{(a) and (b) illustrate eigenvalues rotated through $-\Delta_p$ and $-\Delta$ for the two-phase QWs in Theorem \ref{TWO_PHASE_QW1} and Theorem \ref{TWO_PHASE_QW2}, respectively. In (a), the bold lines indicate the range of the existence of eigenvalues rotated through $-\Delta_p$. In (b), the bold lines indicate the range of the existence of eigenvalues rotated through $-\Delta$, where $j=p\  (|\beta_p|\leq|\beta_m|),\ =m\  (|\beta_m|<|\beta_p|)$. }
\label{Fig:2}
\end{figure}

\begin{remark}
	The model in Theorem \ref{TWO_PHASE_QW2} is an extension of the model of the following previous research:
	\begin{enumerate}
	\item[]$\bullet$ Endo, Konno, Obuse (2015) \cite{EKO} : 
	\[(\alpha_p,\beta_p,\Delta_p)=(\frac{i}{\sqrt{2}}, \frac{ie^{i\sigma_+}}{\sqrt{2}}, \frac{3\pi}{2}),\  (\alpha_m,\beta_m,\Delta_m)=(\frac{i}{\sqrt{2}}, \frac{ie^{i\sigma_-}}{\sqrt{2}}, \frac{3\pi}{2})
	,\ \sigma_{\pm}\in[0,2\pi).\]
	\end{enumerate}
\end{remark}

    \begin{theorem}\label{TWO_PHASE_ONE_DEFECT_QW2}
     We assume $\beta_o=0$, $\ |\beta_p| = |\beta_m|=|\beta|$  and $\Delta_p=\Delta_m =\Delta$. Let $C=\Delta_{0}+(\arg \beta_p - \arg \beta_{m})/2$, $\sigma_p(U)\neq \emptyset$  holds. All eigenvalues of $U$ are given by the following conditions:
     \begin{enumerate}
	\item[]\ $\bullet$\ \rm Condition 1 :\quad  $\sin( \Delta -C) \in [ -1,|\beta |)$.
	\item[]\ $\bullet$\ \rm Condition 2 :\quad $\sin( \Delta -C) \in ( -|\beta |,1]$.
\end{enumerate}
    $\bullet$ If Condition 1 holds, $e^{i\lambda_1 },\  e^{i\lambda_2}\in\sigma_p(U), where $ 
     \begin{align*}
     e^{i\lambda_1 } =\frac{e^{i\Delta } -i|\beta |e^{iC}}{\left| e^{i\Delta } -i|\beta |e^{iC}\right| },\qquad e^{i\lambda_2 }=-e^{i\lambda_1 }.
     \end{align*}
    $\bullet$ If Condition 2 holds, $e^{i\lambda_3}$, $e^{i\lambda_4} \in \sigma_p(U)$, where
     
     \begin{align*}
     e^{i\lambda_3 }=\frac{e^{i\Delta } +i|\beta |e^{iC}}{\left| e^{i\Delta } +i|\beta |e^{iC}\right| } ,\qquad e^{i\lambda_4 }=-e^{i\lambda_3 }.
     \end{align*} In particular, when both of Condition 1 and Condtion 2 hold, $U$ has 4 eigenvalues written above.
    \end{theorem}
    \begin{figure}[H]
	\centering
\includegraphics[width=9cm, height=7cm]{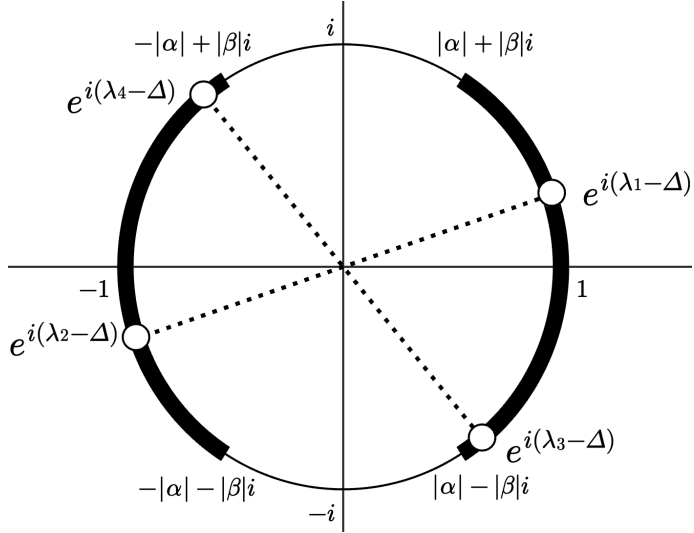}
\caption{The illustration of eigenvalues rotated through $-\Delta$ for the two-phase QWs with one defect in Theorem \ref{TWO_PHASE_ONE_DEFECT_QW2}. The bold lines indicate the range of the existence of eigenvalues rotated through $-\Delta$.  }
\end{figure}

\begin{remark}
The model in Theorem \ref{TWO_PHASE_ONE_DEFECT_QW2} is an extension of the model in following previous research:
	\begin{enumerate}
	\item[] $\bullet$ Endo, Konno, Segawa, Takei (2015) \cite{EKST2} : 
	\begin{align*}
	&(\alpha_p,\beta_p,\Delta_p)=(\frac{i}{\sqrt{2}}, \frac{ie^{i\sigma_+}}{\sqrt{2}}, \frac{3\pi}{2}),
    \\
    &(\alpha_m,\beta_m,\Delta_m)=(\frac{i}{\sqrt{2}}, \frac{ie^{i\sigma_-}}{\sqrt{2}}, \frac{3\pi}{2}),
    \\
	&(\alpha_o,\beta_o,\Delta_o)=(i, 0, \frac{3\pi}{2}),
	\qquad \sigma_{\pm}\in[0,2\pi).
	\end{align*}
	\end{enumerate}
\end{remark}
\section{Summary}
	In this paper, we analysed eigenvalues of two-phase quantum walks with one defect on the integer lattice $\mathbb{Z}$ via the method of transfer matrices.
In Theorem \ref{THEO_1}, we acquired the necessary and sufficient condition for the existence of eigenvalues. By using this condition, we succeeded in clarifying concrete eigenvalues for five models including  previous studies. In Theorem \ref{THEO_QW1} and \ref{THEO_QW2}, we treated two cases of one-defect QWs.
For the model in Theorem \ref{THEO_QW1}, the four eigenvalues exist only if $|\beta|^2>\Re (\beta \overline{\beta_0})$ is satisfied, otherwise they do not.
Moreover, the model in Theorem \ref{THEO_QW2} has two conditions, $|\beta|\cos(\Delta_o-\Delta)\pm|\alpha|\sin(\Delta_o-\Delta)<|\beta|$. Four eigenvalues exist if both conditions are satisfied, and two eigenvalues exist if only one of them is satisfied.
Theorem \ref{TWO_PHASE_QW1} and Theorem \ref{TWO_PHASE_QW2} are results for two-phase QWs, and there are two eigenvalues if and only if $\cos(\Delta_m -\Delta_p)<|\beta_p||\beta_m|-|\alpha_p||\alpha_m|$ and $\left(\mathfrak{R}\left(\beta_{m} \overline{\beta_{p}}\right)-\left|\beta_{p}\right|^{2}\right)\left(\mathfrak{R}\left(\beta_{m} \overline{\beta_{p}}\right)-\left|\beta_{m}\right|^{2}\right)>0$ are satisfied, respectively.
Finally, Theorem \ref{TWO_PHASE_ONE_DEFECT_QW2} is a result for two-phase QWs with one defect.
As in the case of Theorem \ref{THEO_QW2}, there are two conditions $\sin(\Delta -C)\in[-1,\, |\beta|)$ and $\sin(\Delta -C)\in(|\beta|,\, 1]$. Four eigenvalues exist if both conditions are satisfied, and two eigenvalues exist if only one of them is satisfied.

%
%



\end{document}